\documentclass[sigconf]{acmart}

\AtBeginDocument{%
  \providecommand\BibTeX{{%
    \normalfont B\kern-0.5em{\scshape i\kern-0.25em b}\kern-0.8em\TeX}}}

\usepackage{amsmath}        
\usepackage{amsthm}         

\usepackage{amssymb}        

\usepackage{mathtools}      
\usepackage[ruled,vlined,linesnumbered]{algorithm2e}    
\usepackage{algpseudocode}    
\usepackage{graphicx}       
\usepackage{subcaption}     
\usepackage{appendix}       
\usepackage{float}          

\usepackage{multirow}       
\usepackage{booktabs}
\usepackage{lipsum}
\usepackage{wrapfig}
\usepackage{marvosym}       

\DeclareMathOperator{\E}{\mathbb{E}}

\newtheorem{theorem}{Theorem}[section]

\DeclareMathOperator{\xinput}{\mathbf{x}}

\acmSubmissionID{20}

\setcopyright{rightsretained}

\begin{document}
\copyrightyear{2021}
\acmYear{2021}
\acmConference[ACM CHIL '21]{ACM Conference on Health, Inference, and Learning}{April 8--10, 2021}{Virtual Event, USA}
\acmBooktitle{ACM Conference on Health, Inference, and Learning (ACM CHIL '21), April 8--10, 2021, Virtual Event, USA}\acmDOI{10.1145/3450439.3451857}
\acmISBN{978-1-4503-8359-2/21/04}

\title{RNA Alternative Splicing Prediction with  \\Discrete Compositional Energy Network}

\author{Alvin Chan}
\authornote{Both authors contributed equally to this research.}
\email{guoweial001@e.ntu.edu.sg}
\affiliation{%
  \institution{Nanyang Technological University}
  \country{Singapore}
}

\author{Anna Korsakova}
\authornotemark[1]
\email{kors0001@e.ntu.edu.sg}
\affiliation{%
  \institution{Nanyang Technological University}
  \country{Singapore}
}

\author{Yew-Soon Ong}
\email{asysong@ntu.edu.sg}
\affiliation{%
  \institution{Nanyang Technological University}
  \country{Singapore}
}

\author{Fernaldo Richtia Winnerdy}
\email{fernaldo.winnerdy@ntu.edu.sg}
\affiliation{%
  \institution{Nanyang Technological University}
  \country{Singapore}
}

\author{Kah Wai Lim}
\email{kwlim@ntu.edu.sg}
\affiliation{%
  \institution{Nanyang Technological University}
  \country{Singapore}
}

\author{Anh Tuan Phan}
\email{phantuan@ntu.edu.sg}
\affiliation{%
  \institution{Nanyang Technological University}
  \country{Singapore}
}


\begin{abstract}
A single gene can encode for different protein versions through a process called alternative splicing. Since proteins play major roles in cellular functions, aberrant splicing profiles can result in a variety of diseases, including cancers. Alternative splicing is determined by the gene's primary sequence and other regulatory factors such as RNA-binding protein levels. With these as input, we formulate the prediction of RNA splicing as a regression task and build a new training dataset (CAPD) to benchmark learned models. We propose discrete compositional energy network (DCEN) which leverages the hierarchical relationships between splice sites, junctions and transcripts to approach this task. In the case of alternative splicing prediction, DCEN models mRNA transcript probabilities through its constituent splice junctions' energy values. These transcript probabilities are subsequently mapped to relative abundance values of key nucleotides and trained with ground-truth experimental measurements. Through our experiments on CAPD\footnote{Dataset available at: \texttt{https://doi.org/10.21979/N9/FFN0XH}}, we show that DCEN outperforms baselines and ablation variants.\footnote{Codes and models are released at: \texttt{https://github.com/alvinchangw/DCEN\_CHIL2021}}

\end{abstract}




\maketitle

\section{Introduction}
RNA plays a key role in the human body and other organisms as a precursor of proteins. RNA alternative splicing (AS) is a process where a single gene may encode for more than one protein isoforms (or mRNA transcripts) by removing selected regions in the initial pre-mRNA sequence. In the human genome, up to 94\% of genes undergo alternative splicing \citep{Wang2008}. AS not only serves as a regulatory mechanism for controlling levels of protein isoforms suitable for different tissue types but is also responsible for many biological states involved in disease \citep{Tazi2009}, cell development and differentiation \citep{Gallego-Paez2017}. While advances in RNA-sequencing technologies \citep{Bryant2012} have made quantification of AS in patients' tissues more accessible, an AS prediction model will alleviate the burden from experimental RNA profiling and open doors for more scalable \emph{in-silico} studies of alternatively spliced genes.

Previous studies on AS prediction mostly either approach it as a classification task or study a subset of AS scenarios. Training models that can predict strengths of AS in a continuous range and are applicable for all AS cases across the genome would allow wider applications in studying AS and factors affecting this important biological mechanism. To this end, we propose AS prediction as a regression task and curate Context Augmented Psi Dataset (CAPD) to benchmark learned models. CAPD is constructed using high-quality transcript counts from the ARCHS4 database \citep{Lachmann2018}. The data in CAPD encompass genes from all 23 pairs of human chromosomes and include 14 tissue types. In this regression task, given inputs such as the gene sequence and an array of tissue-wide RNA regulatory factors, a model would predict, for key positions on the gene sequence, these positions' relative abundance ($\psi$) in the mRNAs found in a patient's tissue.

Each mRNA transcript may contain one or more splice junctions, locations where splicing occurred in the gene sequence. We hypothesize that a model design that considers these splice junctions would be key for good AS prediction. More specifically, these splice junctions are produced through a series of molecular processes during splicing, so modeling the \text{energy} involved at each junction may offer an avenue to model the whole AS process. This is the intuition behind our proposed discrete compositional energy network (DCEN) which models the energy of each splice junction, \textbf{composes} candidate mRNA transcripts' energy values through their constituent splice junctions and predicts the transcript probabilities. The final component maps transcript probabilities to each known exon start and end's $\psi$ values. DCEN is trained end-to-end in our experiments with ground-truth $\psi$ labels. Through our experiments on CAPD, DCEN outperforms other baselines that lack the hierarchical design to model relationships between splice site, junctions and transcript. To the best of our knowledge, DCEN is the first approach to model the AS process through this hierarchical design. While DCEN is evaluated on AS prediction here, it can potentially be used for other applications where the compositionality of objects (e.g., splice junctions/transcripts) applies. All in all, the prime contributions of our paper are as follows:
\begin{itemize}
    \item We construct Context Augmented Psi Dataset (CAPD) to serve as a benchmark for machine learning models in alternative splicing (AS) prediction.
    \item To predict AS outcomes, we propose discrete compositional energy network (DCEN) to output $\psi$ by modeling transcript probabilities and energy levels through their constituent splice junctions.
    \item Through experiments on CAPD, we show that DCEN outperforms baselines and ablation variants in AS outcome prediction, generalizing to genes from withheld chromosomes and of much larger lengths.
\end{itemize}

\section{Background: RNA Alternative Splicing} \label{sec:background}
RNA alternative splicing (AS) is a process where a single gene (DNA / pre-mRNA) can produce multiple mRNAs, and consequently proteins, increasing the biodiversity of proteins encoded by the human genome. Pre-mRNAs contain two kinds of nucleotide segments, introns and exons. Each post-splicing mRNA transcript would only have a subset of exons while the introns and remaining exons are removed. A molecular machine called spliceosome joins the upstream exon's end with the downstream exon's start nucleotides to form a splice junction and removes the intronic segment between these two sites. In an example of an exon-skipping AS event in Figure~\ref{fig:background}, a single pre-mRNA molecule can be spliced into more than one possible mRNA transcripts ($T_{\alpha}$, $T_{\beta}$) with different probabilities ($P_{T_{\alpha}}, P_{T_{\beta}} \in [0,1]$). These probabilities are largely determined by local features surrounding the splice sites (exon starts/ends) such as the presence of key motifs on the exonic and intronic regions surrounding the splice sites nucleotide. Global contextual regulatory factors such as RNA-binding proteins and small molecular signals \citep{Witten2011,Boo2020,Taladriz-Sender2019} can also influence the transcript probabilities, creating variability for AS outcomes in cells from different tissue types or patients. While exon-skipping is the most common form of AS, there are others such as alternative exon start/end positions and intron retention.

\subsection{Measurement of Alternative Splicing Outcome} One standard way to quantify AS outcome for a group of cells is through \textbf{percent spliced-in} ($\psi \in [0,1]$). Essentially, $\psi$ is defined as the ratio of relative abundance of an exon over all mRNA products within a single gene. An exon with $\psi = 1$ means that it is included in all mRNAs found from experimental RNA-sequencing measurements while $\psi = 0$ means the exon is missing in that particular gene. $\psi$ can also be annotated onto exon's key positions such as its start and end locations. This allows one to approach the AS prediction as a regression task of predicting $\psi$ for each nucleotide of interest.

\begin{figure}[!htbp]
    \centering
    \includegraphics[width=\linewidth]{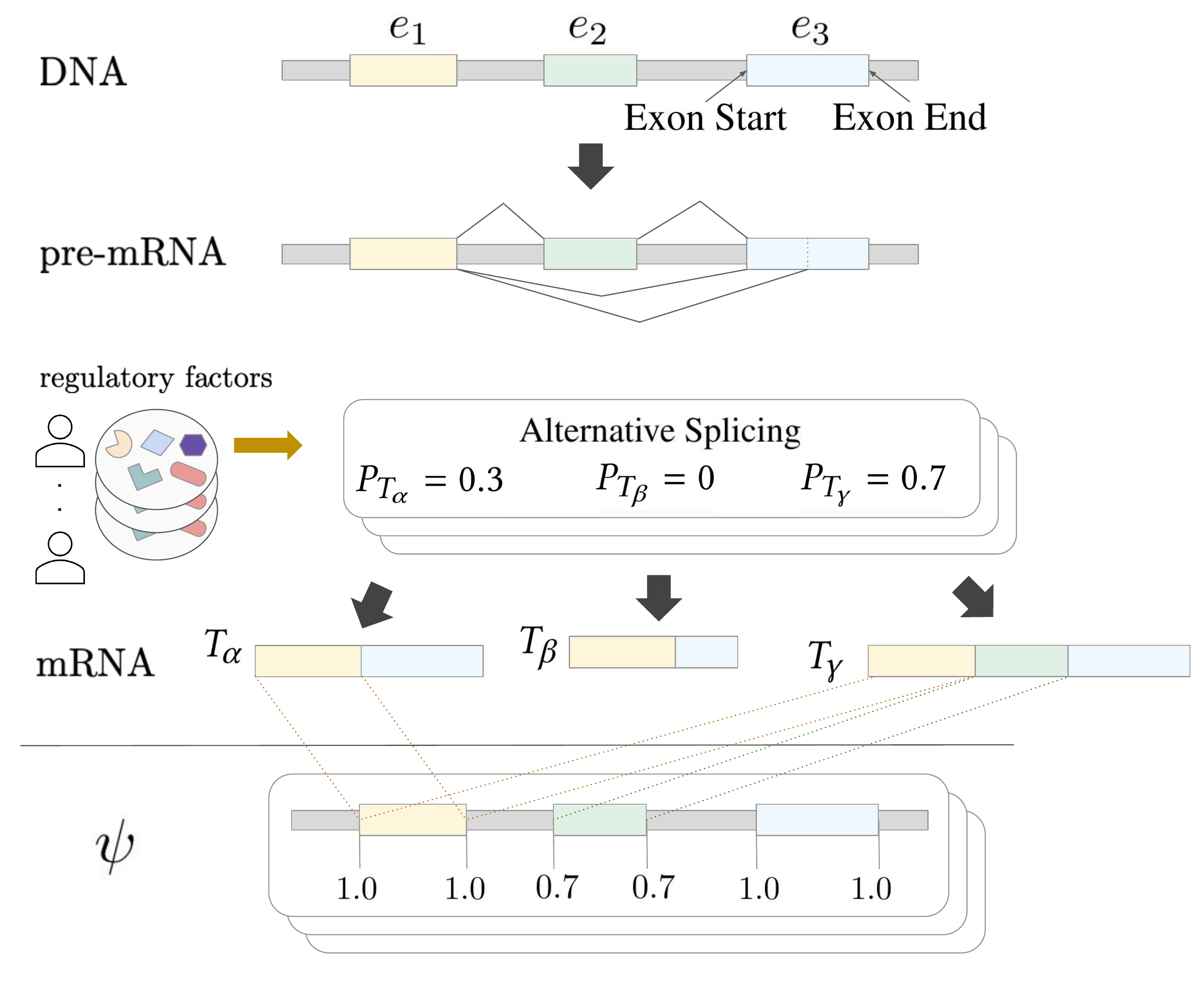}
    \caption{Mechanism of alternative splicing and its relationship with $\psi$ annotations. Introns in the gene sequences are colored gray while exons ($e_1$, $e_2$, $e_3$) are colored otherwise. Each patient sample has an unique set of regulatory factors, leading to a variety of alternative splicing outcomes in the population. In the training of DCEN, the DNA sequences of patients are assumed to be the same.}
    \label{fig:background}
\end{figure}

\section{Related Work}
We review prior art on RNA splicing prediction and energy-based models, highlighting those most similar to our work.

\subsection{Splice Site Classification}
The earliest task of machine learning on RNA splicing involves classification of splicing sites such as exon start and end positions in a given gene sequence, first using models such as decision trees \citep{Pertea2001} and support vector machines \citep{Degroeve2005}. As deep learning gains wider adoption, a line of works uses neural networks for splice site prediction from raw sequence \citep{Zuallaert2018, Zhang2018, Louadi2019, Jaganathan2019}. In a recent example, \cite{Jaganathan2019} used a 1-D Resnet model to classify individual nucleotides in a pre-mRNA sequence into 3 categories: 1) exon's start, 2) exon's end or 3) none of the two classes. Unlike these models that only classify splice sites, we propose DCEN to predict $\psi$ levels of splice sites which involve the consideration of patient-specific input such as levels of RNA regulatory factors on top of just primary gene sequences.

\subsection{Alternative Splicing Prediction}
The prior work in alternative splicing prediction can be categorized into two distinct groups. The first group framed the prediction as a classification task, whether an alternative splicing event would occur given input or change in input. The earliest examples involved using a Bayesian regression \citep{Barash2010} and Bayesian neural network \citep{Xiong2011} to predict whether an exon would be skipped or included in a transcript. \cite{Leung2014} used a neural network with dense layers to predict the type of AS event. Using a classification framework \citep{Xiong2015} to predict one of three classes (high/medium/low), the relative $\psi$ value can also be inferred. Another deep learning-based approach \citep{Louadi2019} utilized a CNN-based framework to classify between four AS event classes (exon skipping, alternative 3', alternative 5' or constitutive exon). 

The second group, which includes DCEN, addresses the prediction as a regression rather than a classification task. This formulation gives higher resolution in the AS event since predicted values correlate with the strength of the AS outcome. \cite{Bretschneider2018} proposed a deep learning model to predict which site is most likely to be spliced given the raw sequence input of 80-nt around the site. The neural networks in \citep{cheng2019mmsplice,cheng2020mtsplice} also use the primary sequences around candidate splice sites as inputs to infer their $\psi$ values. These approaches predict $\psi$ values of splice sites using only their individual representations without modeling the relationship between these splice sites and their parent transcripts, and are conceptually similar to the 2000-nt SpliceAI baseline here which is outperformed by DCEN in our experiments.

Since cellular signals such as RNA-binding proteins (RBPs) are observed to affect RNA splicing \citep{Witten2011, Yee2019}, models such as \cite{jha2017integrative} predicts $\psi$ values with RBP information and generated genomic features used as input. while \cite{Huang2017,Zhang2019} have emerged to incorporate both primary sequence features and RBP levels to better predict exon inclusion levels given a small number of experimental read counts. While also considering regulatory factors such as RBPs, our approach differs from \cite{Witten2011, Yee2019} as we do not assume the availability of experimental read counts for the gene of interest. To the best of our knowledge, DCEN is the first approach to model whole transcript constructs (through energy levels) on top of the immediate neighborhood around the nucleotide of interest when predicting its $\psi$ value in the splicing process.

\subsection{Energy-Based Models}
Most recent work in energy-based models (EBM) \citep{lecun2006tutorial} focused on the application of image generative modeling. Neural networks were trained to assign low energy to real samples \citep{xie2016theory,song2018learning,du2019implicit,nijkamp2019learning,grathwohl2019your} so that realistic-looking samples can be sampled from the low-energy regions of the EBM's energy landscape. Instead of synthesizing new samples, our goal here is to predict RNA splicing outcomes. Other applications of EBMs include anomaly detection \citep{song2018learning}, protein conformation prediction \citep{du2020energy} and reinforcement learning \citep{haarnoja2017reinforcement}. Previous compositional EBMs such as \cite{haarnoja2017reinforcement,du2019compositional} considered high dimensional continuous spaces in their applications which makes sampling from the model intractable. In contrast, since genes consist of a finite number of known transcripts, DCEN considers discrete space where the probabilities of the transcripts are tractable through importance sampling with a uniform distribution.

\section{CAPD: Context Augmented Psi Dataset} \label{sec:data}
The core aim of Context Augmented Psi Dataset (CAPD) is to construct matching pairs of sample-specific inputs and labels to frame the alternative splicing prediction as a regression task and facilitate future benchmarking of machine learning splicing prediction models. Each CAPD sample is a unique AS profile of a gene from the cells of a particular tissue type, from an individual patient. Its annotations contain $\psi \in [0,1]$ of all the know exon starts and ends for the particular gene. Apart from the $\psi$ labels, each data sample also contains the following as inputs: a) full sequence of the gene ($\xinput$), b) nucleotide positions of all the known transcripts ($\mathcal{T}$) on the full gene sequence and c) levels of RNA-regulatory factors ($\xinput_{\text{reg}}$).

\begin{figure*}[!htbp]
    \centering
    \includegraphics[width=\textwidth]{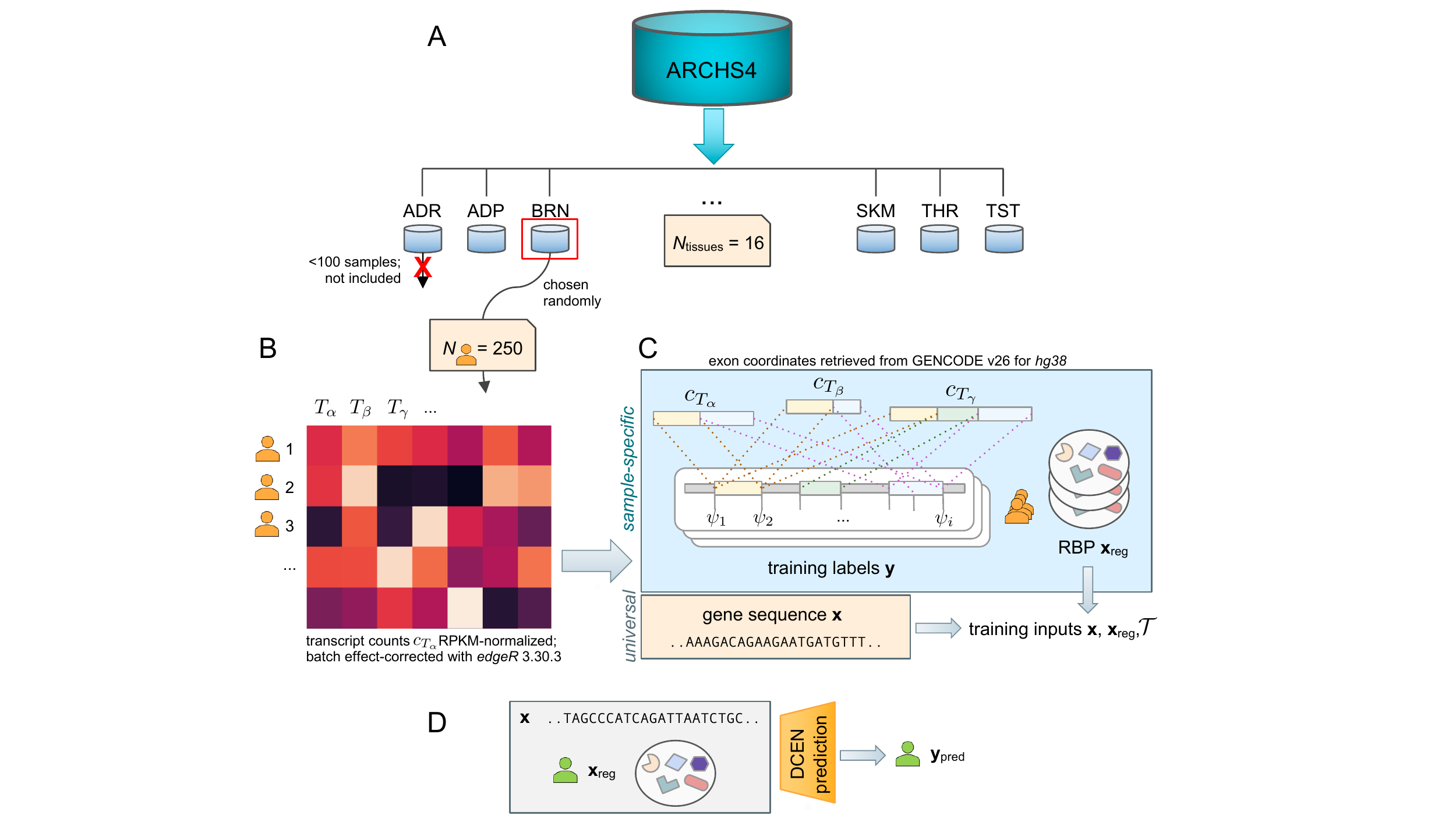}
    \caption{Pipeline of the CAPD dataset preparation with the steps as follows: A) a simple keyword search in the ARCHS4 v.8 database is performed to find a given number of samples for the specified tissues. Tissues for which the number of samples found is less than a specified threshold (N=100) are not included. B) Transcript libraries are normalized to RPKM, processing by \textit{edgeR} library and outlier removal are performed. C) Normalized transcript counts are then used to construct training labels \textbf{y} and auxiliary regulatory vector inputs \textbf{x$_{reg}$} for each sample. The latter is then used together with pre-mRNA transcript sequence as model input. D) After model training, a sample input for $\psi$ prediction consists of pre-mRNA transcript sequence and a vector of abundance of sample-specific regulatory factors.}
    \label{fig:pipeline}
\end{figure*}

\subsection{Construction of CAPD} 
We mine transcript abundance data from the publicly available ARCHS4 database v.8 \citep{Lachmann2018}. ARCHS4 database v.8 contains expression data for 238522 publicly available human RNA-seq samples that were retrieved from Gene Expression Omnibus (GEO) and aligned to human transcriptome Ensembl 90 cDNA \citep{Zerbino2018} to produce count numbers for each transcript in each sample. A simple keyword search was used to find $250 \times N_T$ samples for each of the 16 tissue types ($N_T = 16$) selected for training: adipose tissue, blood, brain, breast, colon, heart, kidney, liver, lung, lymph node, prostate gland, skeletal muscle tissue, testes, thyroid gland. Underrepresented tissues, i.e. with the number of samples below a threshold ($N = 100$) were not included (see the pipeline in Figure~\ref{fig:pipeline}A). Each sample row contains transcript-centric raw RNA counts. Standard normalization of RNA read counts to RPKM \citep{Mortazavi2008} was then performed. To exclude samples with significantly different expression patterns, a z-score outlier removal procedure described in \citep{Oldham2008} for similar tasks was applied to samples from each tissue separately. Expression patterns of samples belonging to one tissue type are varying, and $\approx 2- 5\%$ of the sample population is usually considered outlying by the algorithm. To homogenize the data, batch effects were removed with \textit{edgeR} R library \citep{Robinson2009} (Figure~\ref{fig:pipeline}B). Heterogeneity of the data, however, is still expected as the samples belong to different sources and were chosen randomly. This gives the normalized transcript count values ($c_{T_{\alpha}}$) for each gene transcript ($T_{\alpha}$).

To construct the levels of RNA splicing regulatory factors ($\xinput_{\text{reg}}$), we extract the expression levels of RBPs corresponding to the RBPDB database \citep{Cook2011} and RNA chemically modifying corresponding to \cite{Basturea2013} in a sample-wise manner from the normalized transcript count matrix (Figure~\ref{fig:pipeline}C). This gives a 3971-dimensional $\xinput_{\text{reg}}$. These, together with universal pre-mRNA transcript sequences (same for each sample), compose model input.

To generate primary pre-mRNA transcript sequences, we follow the procedure described in \cite{Jaganathan2019}: pre-mRNA (synonymous to DNA, with T $\to$ U) sequences are extracted with flanking ends of 1000 nt on each side, while intergenic sequences are discarded. Pseudogenes, genes with sequence assembly gaps and genes with paralogs are excluded from the data. Exon coordinate information is retrieved from GENCODE Release 26 for GRCh38 \citep{Frankish2019} comprehensive set, downloaded from the UCSC table browser \footnote{https://genome.ucsc.edu/cgi-bin/hgTables}. We omit genes with missing matching GENCODE ID, resulting in a total of 19399 unique human gene sequences. The coordinates of the gene pre-mRNA sequence start and end are determined by the left- or right-most position among all the transcripts for that gene, further extended with flanking ends of 1000 nt for context on each side. The CAPD dataset is split into train and test according to the chromosome number and length of the genes. All genes in chromosome 1, 3, 5, 7, and 9 are withheld as test samples (Test-Chr), similar to \cite{Jaganathan2019}. To test that models can generalize to genes of longer lengths, we further withhold all genes from the other chromosomes that have $>$100K nt (Test-Long) and group them in the test set. The remaining genes are used for training. Key statistics of the CAPD is summarized in Table~\ref{tab:CAPD stats}.

\subsection{Label Annotation} 
The $\psi$ labels are constructed as follows (Figure~\ref{fig:pipeline}C): 1) the count values $c_i$ for all known exon start and end (position $i$) are initialized to zero. 2) enumerating through all transcripts, each transcript count is added onto the counts of its constituent exons' start/end  $\mathbf{c}_i \gets \mathbf{c}_i + c_{T_{m}}, \forall i \in T_{m}$. 3) To compute $\psi$ values, each count value is divided by the sum of transcript counts to normalize its value to $[0,1]$, i.e., $\psi_i = \mathbf{c}_i / \sum_m c_{T_{m}}$ .

\begin{table}[!htbp]
    \centering
    \caption{CAPD data statistics.}
        \begin{tabular}{ l|ccc }
         ~ & Train & Test-Chr & Test-Long \\
         \hline
         \# of unique genes & 11,472 & 5,604 & 2,323 \\
         \hline
         mean pre-mRNA length (nt) & 26,196 & 73,355 & 247,041 \\
         \hline
         mean \# of exons & 6.7 & 7.0 & 10.5 \\
        \end{tabular}
\label{tab:CAPD stats}
\end{table}


\section{DCEN: Discrete Compositional Energy Network}
In \S~\ref{sec:background}, we learn that final mRNA splice isoforms may comprise one or more splice junctions, points where upstream exon's end and downstream exon's start meet. Inspired by the creation of splice junctions by spliceosomes, the first stage of DCEN models the energy values of the splicing process at splice junctions. As the splice junctions are typically far enough ($\sim$300 nt) and assumed to be independent of one another, the energy of a final mRNA transcript can be composed by the summation of its constituent splice junctions' energy. The second stage of DCEN derives probabilities for the formation of each transcript from their energy values. These transcript probabilities are then mapped to the relative abundance ($\psi$) of exon starts/ends at its corresponding splice sites. Since a particular splice junction may appear in more than one mRNA isoform, we design DCEN to be invariant to the splice junctions. In the following \S~\ref{sec:model architecture}, we discuss the key components of DCEN while \S~\ref{sec:training algorithm} details its training process.

\subsection{Model Architecture} \label{sec:model architecture}
Here, we detail the DCEN learned energy functions and how transcript probabilities can be derived from their energy levels through Boltzmann distribution and importance sampling. A summary of DCEN model architecture is shown in Figure~\ref{fig:DCEN}.



\begin{figure}[!htbp]
    \centering
    \includegraphics[width=\linewidth]{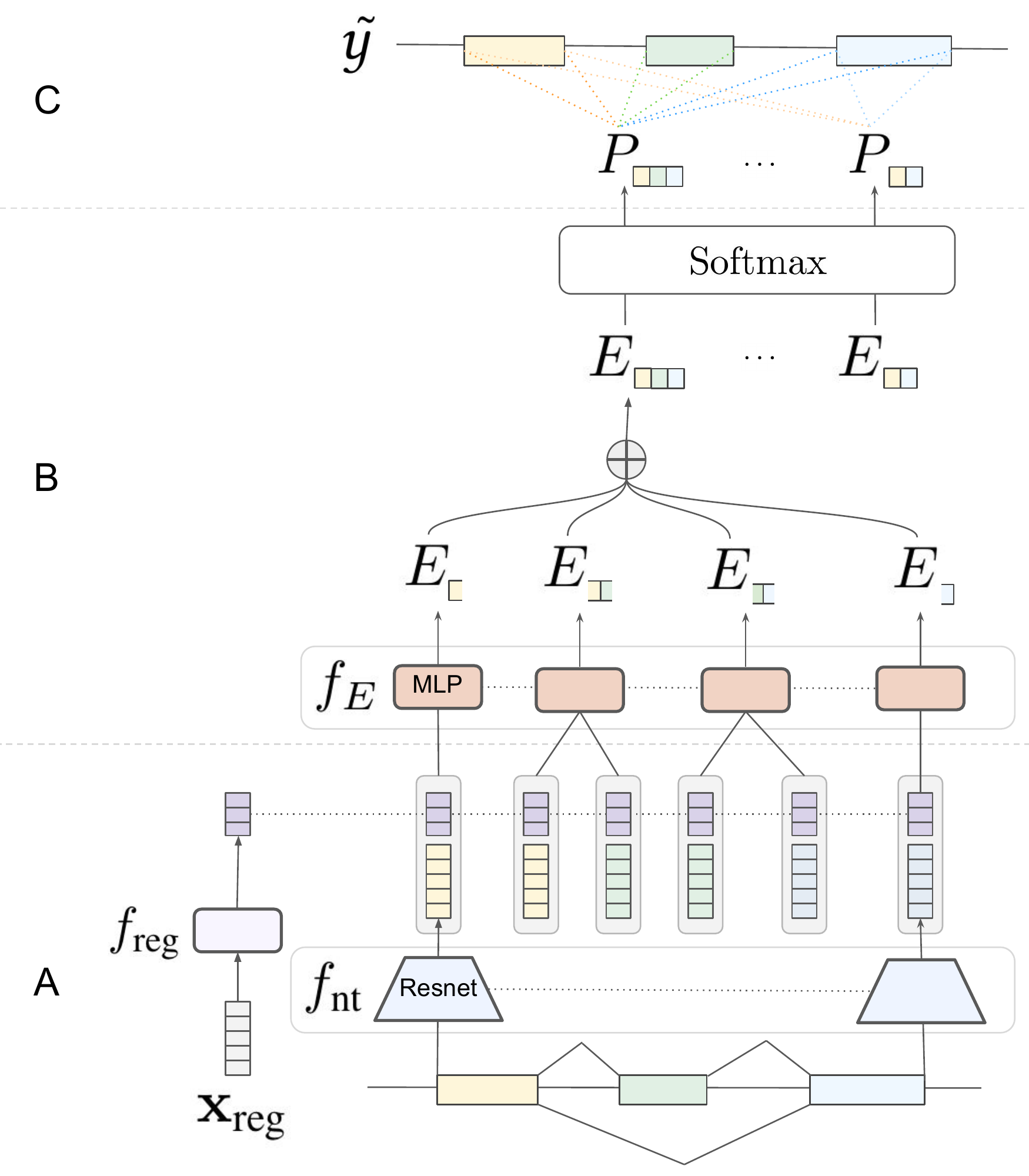}
    \caption{Model architecture of proposed discrete compositional energy network (DCEN). A) Learned representations of the DNA sequence are extracted by $f_{\text{nt}}$ while the representations of the regulatory factors by $f_{\text{reg}}$. B) The representation of a particular splice junction is the concatenation of acceptor and donor sites' representations. The learned energy value of a particular splice junction is computed by passing its representation through $f_{E}$. The energy level of each transcript is the sum of its junctions' energies. C) Through a softmax operation, the probability of each transcript in a gene is computed. By pooling the probabilities of all transcripts that contain a particular nucleotide, we get the nucleotide's predicted $\psi$ value.}
    \label{fig:DCEN}
\end{figure}

\subsection{Learned Energy Functions} \label{sec:transcript probabilities from energy levels}
The weights of DCEN's learned energy functions consist of a 1) feature extractor, 2) regulatory factors encoder and 3) junction energy network. The feature extractor ($f_{\text{nt}}$) takes the pre-mRNA sequence of length $l$ as its input ($\xinput \in \mathbb{R}^{l \times 4}$) where 4 is the number of possible nucleotides and outputs a hidden representation ($ \mathbf{h}_{\text{nt}} \in \mathbb{R}^{l \times d}$) for each nucleotide position ($\xinput_i$) while the regulatory factors encoder ($f_{\text{reg}}$) takes in the levels of regulatory factors ($\xinput_{\text{reg}}$) to compute a gene-wide hidden states $\mathbf{h}_{\text{reg}} \in \mathbb{R}^d$:

\begin{equation}
    \mathbf{h}_{\text{nt}} = f_{\text{nt}}(\xinput) ~, ~~~~~ \mathbf{h}_{\text{reg}} = f_{\text{reg}} (\xinput_{\text{reg}})
\end{equation}

We concatenate $\mathbf{h}_{\text{reg}}$ to all the position-specific $\mathbf{h}_{\text{nt}}$ to form a new position-wise hidden state ($\mathbf{h}_i$) that is dependent on regulatory factors. The representation of a particular splice junction ($J_{k} = (i,j)$) is the concatenation between the hidden states of upstream exon end's and downstream exon start's hidden states:


\begin{equation}
    \mathbf{h}_{J_{k}} = [ \mathbf{h}_i ; \mathbf{h}_j ] ~~, ~~~~~ \mathbf{h}_i = [ \mathbf{h}_{\text{nt}_i} ; \mathbf{h}_{\text{reg}} ]
\end{equation}

If an exon start/end is the first/last nucleotide of a transcript, its hidden state is concatenated with a learned start/end token instead ($\mathbf{h}_\text{start}$ or $\mathbf{h}_\text{end}$ respectively). To model the energy ($E_{J_{k}} \in \mathbb{R}$) of producing splice junction $J_{k}$, we feed its representation into the energy network ($f_{E}$).
We sum up the energy values of all splice junctions ($J_{k}$) inside a mRNA transcript ($T_{\alpha}$) to compose the total energy ($E_{T_{\alpha}} \in \mathbb{R}$) involved in producing the transcript from a splicing event: 
\begin{equation}
    E_{T_{\alpha}} = \sum_k E_{J_{k}} ~~, ~~~~~ \forall J_{k} \in T_{\alpha} ~~, ~~~~~ E_{J_{k}} = f_{E} (\mathbf{h}_{J_{k}})
\end{equation}

\subsection{Transcript Probabilities from Energy Values} \label{sec:transcript probabilities from energy levels}
After obtaining the energy levels ($E_{T_{\alpha}}$) of all mRNA transcript candidates for a particular gene, we can compute the probabilities of these transcripts via a softmax operation through the theorem below.

\begin{theorem} \label{theorem:energy levels to probabilities}
Given the energy levels of all the possible discrete states of a system, the probability of a particular state $T_i$ is the softmax output of its energy $E_{T_i}$ with respect to those of all other possible states in the system, i.e.,

\begin{equation} \label{eq:energy levels to probabilities}
P_i = \frac{ \exp (- E_{T_i} ) }{ \sum_{j} \exp (- E_{T_j}  ) } = \mathrm{Softmax}_i( E )
\end{equation}

\end{theorem}

Its proof, deferred to the Appendix \S~\ref{sec:proof}, can be derived through Boltzmann distribution and importance sampling. Since each mRNA transcript $T_{\alpha}$ can be interpreted as a discrete state of the alternative splicing event for a particular gene (system) in Theorem~\ref{theorem:energy levels to probabilities}, we can compute its probability from its energy value. 

It is important to also consider a null state with energy $E_{\text{null}}$ where none of the gene's mRNA transcripts is produced. In DCEN, $E_{\text{null}}$ is a learned parameter. In summary, the probability of producing transcript $T_{\alpha}$ in an gene splicing event is 
\begin{equation}
    P_{T_{\alpha}} = \mathrm{Softmax}_{\alpha} (E) ~~, ~~~~~ E = [E_{T_{\alpha}}, E_{T_{\beta}} , \dots, E_{\text{null}} ] 
\end{equation}

If the null state were not considered, the model would incorrectly assume that a particular gene is always transcribed since the sum of transcripts' probabilities in a gene would be $\sum_{i} P_{T_i} = 1$.

\subsection{Exon Start/End Inclusion Levels from Transcript Probabilities} \label{sec:psi from transcript probs}
We can compute the probability of a particular (exon start/end) nucleotide of position $i$ by summing up the probabilities of all transcripts that contain that nucleotide.
\begin{equation}
 P_i = \sum_m P_{T_{m}} ~~, ~~ \forall T_{m} \in \mathcal{T}_{i} 
\end{equation}
where $\mathcal{T}_{i}$ is the set of transcripts containing the nucleotide of interest. By inferring the (exon start/end) nucleotide inclusion levels through transcript probabilities, our model has the advantage of additional access to the relative transcription level of the gene's transcripts over baselines that infer directly at the nucleotide positions (such as baselines in \S~\ref{sec:ablation}).

\subsection{Training Algorithm} \label{sec:training algorithm}
\subsubsection{Regression Loss} 
Since experimental $\psi$ levels from CAPD are essentially the empirical observations of nucleotide present in the final mRNA transcript products, its normalized values ($y_i \in [0,1]$) can be used as the ground-truth label for the predicted nucleotide inclusion levels. This allows us to train DCEN as a regression task by minimizing the mean squared error (MSE) between the predicted nucleotide $\psi$ and the normalized experimental $\psi$ values:

\begin{equation}
    L_{\psi} = \sum_i \| \tilde{y_i} - y_i \|^2_2 ~~, ~~~~~ \tilde{y_i} = P_i
\end{equation}

In our experiments, only nucleotide positions that are either an exon start or end are involved in this regression training objective.

\subsubsection{Classification Loss} 
We also include a classification objective, similar to \cite{Jaganathan2019}, where a classification head $f_{\text{cls}}$ takes the nucleotide hidden states ($\mathbf{h}_{\text{nt}}$) as input to predict probability of every nucleotide as one of the 3 classes (exon start, end and neither) to give the classification loss:
\begin{equation}
    L_{\text{cls}} = - \mathbf{y}_{\text{cls}}^\top \log f_{\text{cls}} (\mathbf{h}_{\text{nt}})
\end{equation}

where $\mathbf{y}_{\text{cls}}$ is the ground-truth labels for each nucleotide. This helps the DCEN learn features on the gene primary sequence that are important for RNA splicing. A summary of the training phase is shown in Algorithm~\ref{algo:DCEN Training}.


\begin{algorithm}[!htbp]
 \caption{Discrete Compositional Energy Network Training}
 \label{algo:DCEN Training}

\textbf{Input:} Training data $\mathcal{D}_{\text{train}}$, Learning rate $\gamma$, 

\For{ each training iteration }{
 Sample $(\xinput, \xinput_{reg}, \mathbf{y}, \mathcal{T}, \mathcal{J} ) \sim \mathcal{D}_{\text{train}}$

 $\mathbf{h}_{\text{nt}} \gets f_{\text{nt}} (\xinput) , ~~$
 
 $L_{\text{cls}} \gets - \mathbf{y}_{\text{cls}}^\top \log f_{\text{cls}} (\mathbf{h}_{\text{nt}})$ \algorithmiccomment{ Compute exon start/end classification cross-entropy loss}
  
 $\mathbf{h}_{\text{reg}} \gets f_{\text{reg}} (\xinput_{reg}) $
 
 $\mathbf{h}_i \gets [ \mathbf{h}_{\text{nt}_i} ; \mathbf{h}_{\text{reg}} ] ~~$
 
 $\mathbf{h}_{J_{k}} \gets [ \mathbf{h}_i ; \mathbf{h}_j ] ~~,~~~$
 $J_{k} = (i,j)~~$ \algorithmiccomment{ Get junction state from upstream exon end $i$ and downstream exon start $j$}
 
 $E_{J_{k}} \gets f_E (\mathbf{h}_{J_{k}}) ~~$ \algorithmiccomment{ Compute splice junction energy}
 
 $E_{T_{\alpha}} \gets \sum_k E_{J_{k}} ~~, ~~ \forall J_{k} \in T_{\alpha}, ~~ \forall T_{\alpha} \in \mathcal{T} $ \algorithmiccomment{ Compose transcript energy}
 
 $P_{T_{\alpha}} \gets \mathrm{Softmax}_{\alpha} (E) ~~, ~~ E = [E_{T_{\alpha}}, E_{T_{\beta}} , \dots ] $ \algorithmiccomment{ Compute transcript probabilities}
 
 $\tilde{y_i} \gets \sum_m P_{T_{m}} ~~, ~~ \forall T_{m} \in \mathcal{T}_{i} $
 
 $L_{\psi} \gets \sum_i \| \tilde{y_i} - y_i \|^2_2 ~~ $ \algorithmiccomment{ Compute exon start/end inclusion regression loss}
 
 $\theta \gets \theta + \gamma ~ \nabla_{\theta} (\lambda_{\text{reg}} L_{\psi} +  \lambda_{\text{cls}} L_{\text{cls}} ) $ 
}
\end{algorithm}


\section{Experiments}
We evaluate DCEN and baselines on the prediction of the $\psi$ values of exon starts and ends in our new CAPD dataset. In the following, we describe the baseline models and DCEN ablation variants before discussing the experimental setup and how well these models generalize to withheld test samples.

\subsection{Baselines and Ablation Variants} 
\label{sec:ablation}
SpliceAI is a 1D convolutional Resnet \citep{he2016deep} trained to predict splice sites on pre-mRNA sequences. We train three variants of SpliceAI to compare as baselines in our experiments: the first (\textbf{SpliceAI-cls}, Figure~\ref{fig:spliceai_cls}) is trained only on the classification objective, similar to the original paper, to predict whether a nucleotide is an exon start, end or neither of them. The second (\textbf{SpliceAI-reg}, Figure~\ref{fig:spliceai_reg_ml}) is trained only on a regression objective like DCEN to directly predict the $\psi$ levels of nucleotides while the third variant (\textbf{SpliceAI-cls+reg}) is trained on both the classification and regression objectives. Both SpliceAI-reg and SpliceAI-cls+reg also have a regulatory factors encoder ($f_{\text{reg}}$) similar to DCEN's to compute $\mathbf{h}_i$ and a regression head ($f_{\psi}$) to output $\psi$ level for each nucleotide. For a direct comparison, DCEN's feature extractor $f_{\text{nt}}$ takes the same architecture as the SpliceAI Resnet. 
Two DCEN ablation variants are also evaluated: The \textbf{Junction-psi} model (Figure~\ref{fig:junction_psi}) predicts the psi levels of a particular splice junction directly rather than its energy level in the case of DCEN. A simpler ablation variant (\textbf{SpliceAI-ML} or SpliceAI-match layers, Figure~\ref{fig:spliceai_reg_ml}) substitutes DCEN's $f_E$ with a position-wise feedforward MLP containing the same number of parameters to verify that DCEN's better performance is not due to more learned parameters.



\subsection{Data \& Models}
We use the CAPD dataset (\S~\ref{sec:data}) for the training and evaluation of all models. 10\% of the CAPD training genes are randomly selected as the validation set for early-stopping while the rest are used as training samples for the models. For DCEN's $f_{\text{nt}}$ and the SpliceAI baselines, we follow the same setup as the SpliceAI-2K model in \cite{Jaganathan2019} which is a Resnet made up of 1-D convolutional layers with a perceptive window of 2K nucleotides, 1K on each flanking sides. The SpliceAI Resnet model has a total of 12 residual units and hidden states of size 32. The number of channels in $\mathbf{h}_{\text{reg}}$ and $\mathbf{h}_{\text{nt}}$ are 32 while $\mathbf{h}$ has a size of 64 channels. We use a 3-layer MLP for $f_{\text{reg}}$. 
The regression head $f_{\psi}$ in SpliceAI-reg and SpliceAI-cls+reg is a 3-layer MLP with a sigmoid activation to outputs a scalar $\psi$ value. DCEN's $f_E$ is a 4-layer MLP and outputs a scalar energy value for each splice junction. Intermediate hidden states of $f_{\text{reg}}$ and $f_E$ all have dimension of 32. 

\subsection{Training \& Evaluation}
All models are trained with Adam optimizer with a learning rate of 0.001 in our experiments. Due to the data's large size, the training is early-stopped when the model's validation performances plateau: less than 25\% of the full train dataset for all models in our experiments. For the training of SpliceAI models, samples are fed into the model with a batch size of 8 sequences with a maximum length of 7K nucleotides (5K labeled + 2K flanking). For training of DCEN and its ablation variants, a SpliceAI-cls\&reg model pretrained on CAPD was used as the weights of $f_{\text{reg}}$, $f_{\text{nt}}$ and only the parameters of $f_E$ is trained to reduce training time. In each training iteration of DCEN and its ablation variants, a batch of 16 genes was used to train the weights. We evaluate all the models on withheld test samples with two standard regression metrics: Spearman rank correlation and Pearson correlation. Pearson correlation measures the linear relationship between the ground-truth and predicted exon start/end inclusion levels while Spearman rank correlation is based on the ranked order of the prediction and ground-truth values.

\begin{figure*}[!htbp]
\begin{subfigure}{0.3\textwidth}
  \centering
  \includegraphics[width=\textwidth]{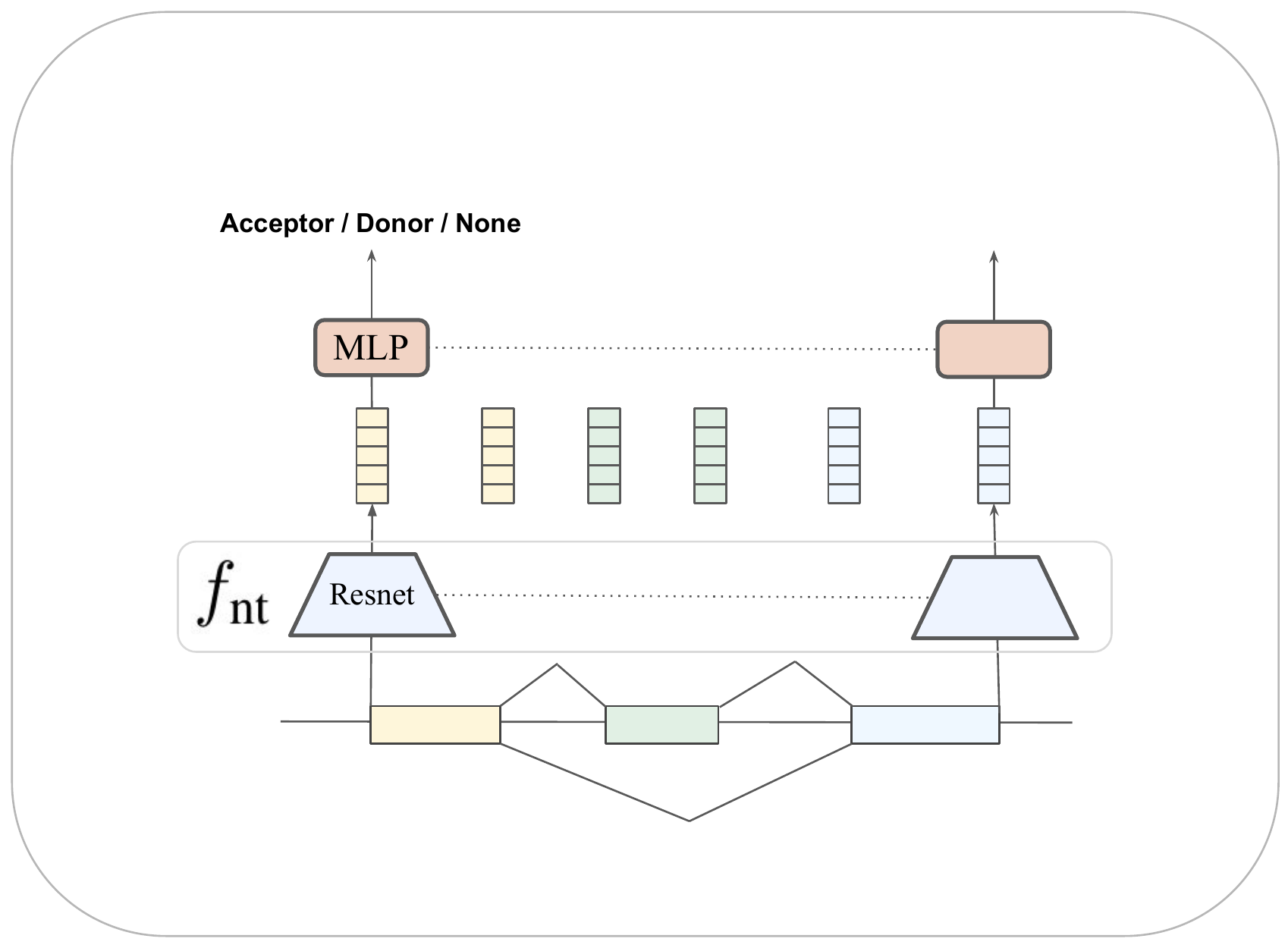}
  \caption{SpliceAI-cls}
  \label{fig:spliceai_cls}
\end{subfigure}
\begin{subfigure}{0.3\textwidth}
  \centering
  \includegraphics[width=\textwidth]{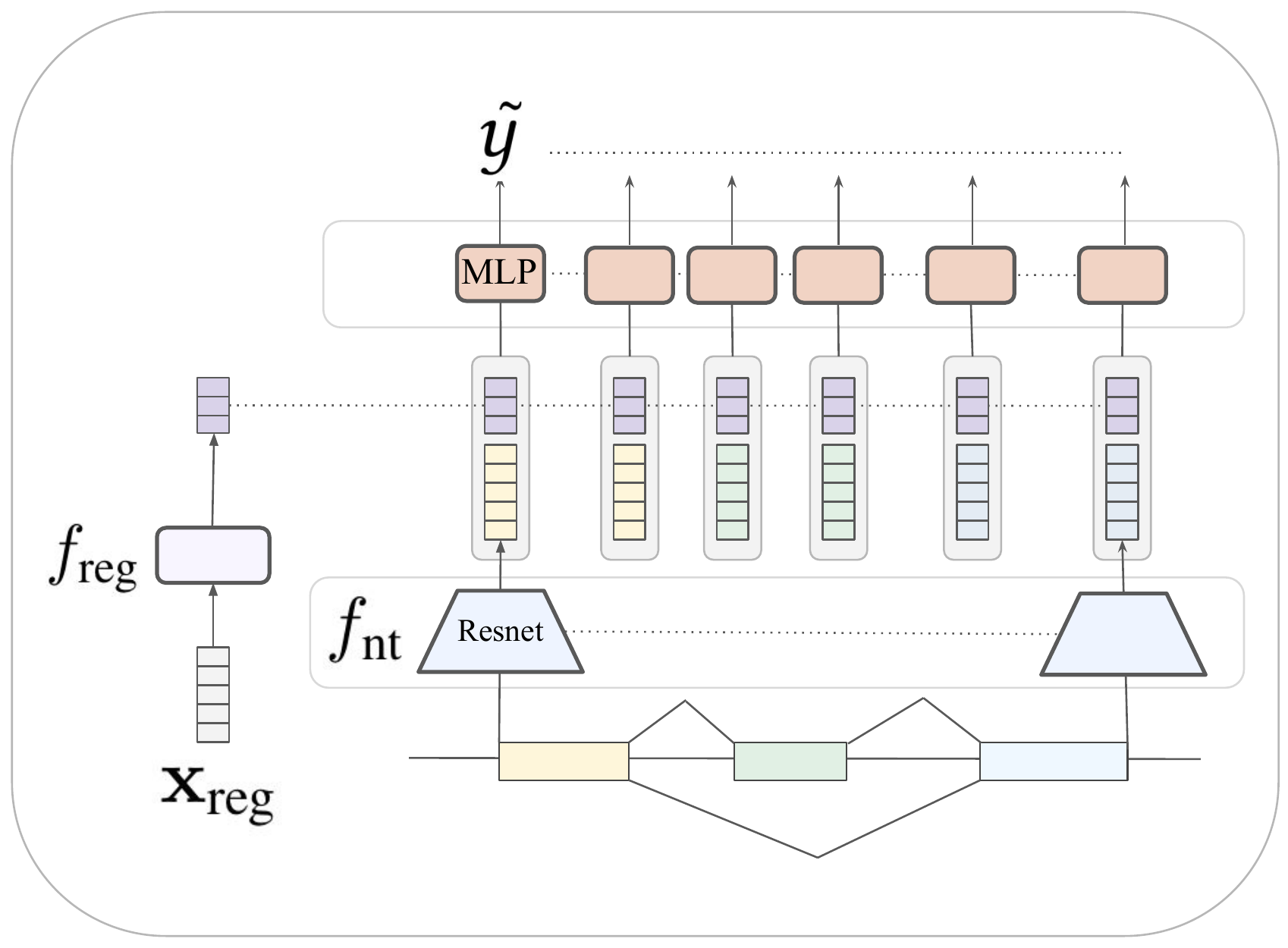}
  \caption{SpliceAI-reg \& -ML}
  \label{fig:spliceai_reg_ml}
\end{subfigure}
\begin{subfigure}{0.3\textwidth}
  \centering
  \includegraphics[width=\textwidth]{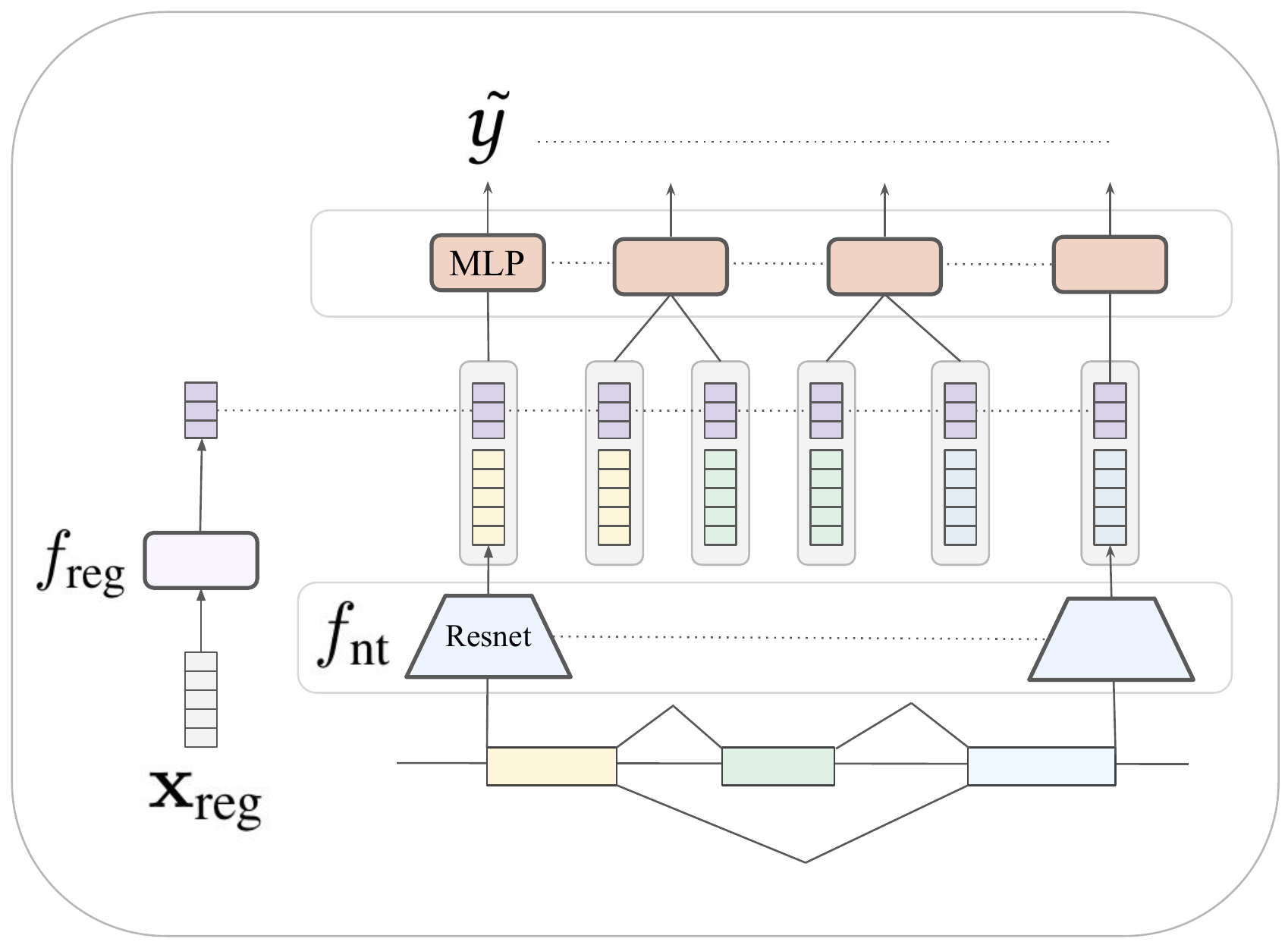}
  \caption{Junction-psi}
  \label{fig:junction_psi}
\end{subfigure}
\caption{Ablation models used to compare with our proposed DCEN model. (a) SpliceAI-cls \cite{Jaganathan2019} is trained to predict whether a particular nucleotide in the DNA sequence is an acceptor, donor site or neither of them from only the primary DNA sequence, without input of regulatory factors ($\xinput_{reg}$). The predicted probability of the acceptor/donor class is assumed to be predicted $\psi$ values during the comparison. (b) SpliceAI-reg and SpliceAI-ML are both trained to predict the $\psi$ value of a particular nucleotide with a regression objective. These two ablation models, unlike DCEN (Figure~\ref{fig:DCEN}), infer $\psi$ values from individual nucleotide's representation and does not rely on the compositionality of transcripts and their splice junctions. (c) Junction-psi infers the $\psi$ value of a particular splice junction from the joint representations of its acceptor and donor sites, without modeling the relationship between junctions and their parent transcripts. }
\label{fig:ablation models}
\end{figure*}

\subsection{Results} 
\label{sec:results}
The DCEN outperforms all baselines and ablation variants for both regression metrics when evaluated on the withheld test samples, as shown in Table~\ref{tab:test}. Even with the same number of parameters, the DCEN shows better performance than the SpliceAI-ML and Junction-psi model. Even with more trained parameters, Junction-psi model performs worse than the SpliceAI-reg and SpliceAI-cls+reg baselines while the SpliceAI-ML does not show clear improvement over these two baselines. 

A key difference between these baselines and DCEN lies in DCEN's inductive bias that models the hierarchical relationships between splice sites and junctions, as well as between splice junctions and their parent transcripts. Even with a matching number of parameters as DCEN, SpliceAI-ML still underperforms DCEN indicating DCEN's model size is not the key contributor to its performance. These observations indicate that DCEN's design to compose transcripts' energy through splice junctions and infer their probabilities from energy values is key for better prediction. The correlation results reported in Table~\ref{tab:test}, \ref{tab:test diff chromosomes} and \ref{tab:test long genes} all have p-values of zero in working precision due to the large number of gene samples.

Training DCEN for 5 times the training steps at the early stopped point does not result in better (Table~\ref{tab:test 5x iters} in S~\ref{sec:Additional results}) nor much worse performance on the test samples. This suggests that DCEN is underfitting the training dataset and may benefit from a larger model size. 

\subsubsection{Excluded chromosomes}
When evaluated separately test samples from chromosomes (1, 3, 5, 7, and 9) not seen during the training phase, DCEN maintains its superior performance (Table~\ref{tab:test diff chromosomes}) when compared to the ablation baselines, showing that it generalizes across novel gene sequences.


\begin{table}[!htbp]
    \centering
    \caption{Performance of DCEN and baselines on all withheld test samples.}
        \begin{tabular}{ lcc }
         \hline
         Model & Spearman Cor. & Pearson Cor. \\
         \hline
         SpliceAI-cls & 0.399 & 0.349  \\
         SpliceAI-reg & 0.579 & 0.574 \\
         SpliceAI-cls+reg & 0.577 & 0.571 \\
         \hline
         SpliceAI-ML & 0.572 & 0.575  \\
         Junction-psi & 0.559 & 0.510  \\
         \hline
         DCEN (ours) & \textbf{0.623} & \textbf{0.651}  \\
         \hline
        \end{tabular}
\label{tab:test}
\end{table}

\subsubsection{Long gene sequences}
DCEN was trained only on genes sequences of length less than 100K nucleotides. From Table~\ref{tab:test long genes}, we observe that DCEN still outperforms the other baselines by a substantial margin and retains most of its performance on these samples when evaluated on genes with long sequences ($>$100K nucleotides). Compared to shorter introns in genes of shorter length, splicing of pre-mRNA with very large introns was observed to occur in a more nested and sequential manner \citep{Sibley2015,Suzuki2013}. Since genes with long sequences contain more large introns, this difference in the splicing mechanism may explain the slight drop in DCEN's performance on genes with longer sequences.


\begin{table}[!htbp]
    \centering
    \caption{Performance of DCEN and baselines on Test-Chr, test samples from chromosomes (1, 3, 5, 7, and 9) different from the training set.}
        \begin{tabular}{ lcc }
         \hline
         Model & Spearman Cor. & Pearson Cor. \\
         \hline
         SpliceAI-cls & 0.400 & 0.353 \\
         SpliceAI-reg & 0.579 & 0.586 \\
         SpliceAI-cls+reg & 0.577 & 0.583 \\
         \hline
         SpliceAI-ML & 0.570 & 0.587 \\
         Junction-psi & 0.560 & 0.525 \\
         \hline
         DCEN (ours) & \textbf{0.622} & \textbf{0.665} \\
         \hline
        \end{tabular}
\label{tab:test diff chromosomes}
\end{table}

\subsubsection{Performance differs across tissue types} \label{sec:tissue type}
We observe that the performance of DCEN's prediction differs across test samples of different tissues types (Table~\ref{tab:test tissue}). The best performing tissue (heart) has Spearman rank correlation of 0.694 and Pearson correlation of 0.716 while the lowest (testes) achieves 0.415 and 0.423 respectively.

\begin{table}[!htbp]
    \centering
    \caption{DCEN tissue-specific performance.}
        \begin{tabular}{ lcc }
         \hline
         Tissue & Spearman Cor. & Pearson Cor. \\
         \hline
        Testes & 0.415 & 0.423 \\
        Brain  & 0.497 & 0.582 \\
        Lung  & 0.585 & 0.624 \\
        Lymph  & 0.586 & 0.621 \\
        Liver  & 0.596 & 0.619 \\
        Blood  & 0.598 & 0.619 \\
        Colon  & 0.616 & 0.645 \\
        Breast  & 0.626 & 0.656 \\
        Adipose  & 0.653 & 0.674 \\
        Muscle  & 0.648 & 0.678 \\
        Kidney  & 0.654 & 0.684 \\
        Prostate  & 0.658 & 0.688 \\
        Thyroid  & 0.674 & 0.691 \\
        Heart  & 0.694 & 0.716 \\
         \hline
        \end{tabular}
\label{tab:test tissue}
\end{table}


\begin{table}[!htbp]
    \centering
    \caption{Performance of DCEN and baselines on Test-Long, withheld samples with long gene sequences ($>$100K nucleotides).}
        \begin{tabular}{ lcc }
         \hline
         Model & Spearman Cor. & Pearson Cor. \\
         \hline
         SpliceAI-cls & 0.400 & 0.344 \\
         SpliceAI-reg & 0.585 & 0.564 \\
         SpliceAI-cls+reg & 0.582 & 0.561 \\
         \hline
         SpliceAI-ML & 0.577 & 0.565 \\
         Junction-psi & 0.560 & 0.495 \\
         \hline
         DCEN (ours) & \textbf{0.626} & \textbf{0.639} \\
         \hline
        \end{tabular}
\label{tab:test long genes}
\end{table}

\begin{figure}[!htbp]
    \centering
    \includegraphics[width=\linewidth]{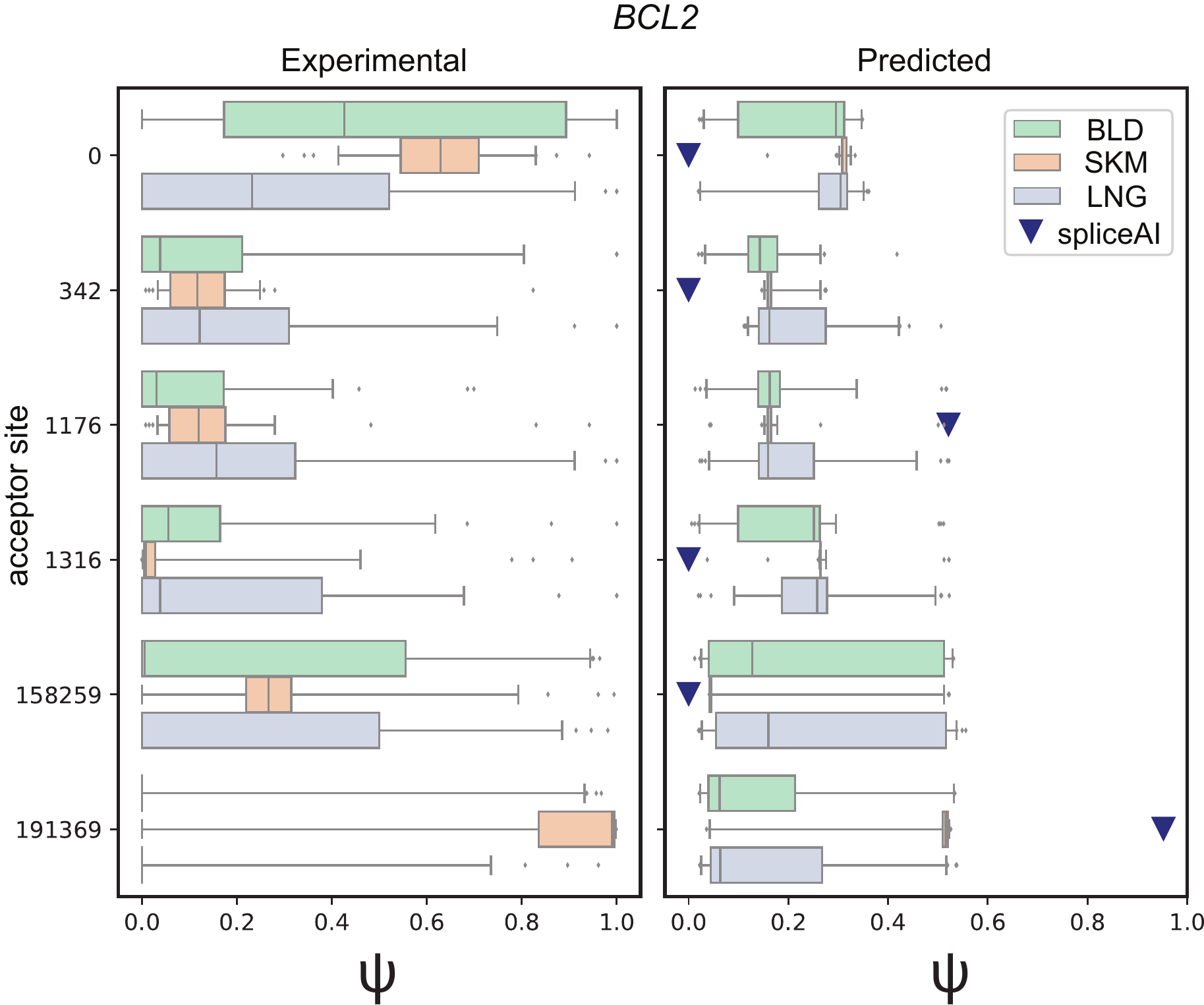}
    \caption{Comparison of the BCL2 gene acceptor sites' $\psi$, experimental vs predicted. Left: box plots of experimental values (grounds truths) of BCL2 acceptor sites in three different tissues: blood, skeletal muscle and lung, denoted as BLD, SKM and LNG respectively. Right: box plots of predicted $\psi$ for the same group of samples using DCEN; SpliceAI-cls prediction (black triangles) shown for comparison. Please note that SpliceAI-cls prediction uses sequence information only and the prediction will be the same for every sample. Acceptor site coordinates (y-axis) are counted relative to the genomic coordinate of the first exon start of BCL2. Data from 50 samples from each tissue are shown.}
    \label{fig:BCL2}
\end{figure}

\subsubsection{Integration of regulatory features allows for sample-specific predictions}
An example of a set of model predictions in comparison with the ground truth values for BCL2 gene acceptor sites is shown in Figure~\ref{fig:BCL2}. Box plots of experimental values (left of Figure~\ref{fig:BCL2}) represent distributions of $\psi$ in 50 samples for three types of tissues: blood, skeletal muscle and lung. The predictions of DCEN vs SpliceAI-cls for these acceptor sites are shown on the right. We observe that DCEN's predictions can partially recover the variances in $\psi$ values unique to tissue types (e.g., acceptor site 158259 in Figure~\ref{fig:BCL2}). When comparing the variance of $\psi$ prediction within each tissue type with the variance of ground-truth $\psi$, the Pearson correlation is 0.177 while the Spearman rank correlation is 0.198. Figure~\ref{fig:BCL2don} in \S~\ref{sec:Additional results} shows the comparison for BCL2 donor sites.


\section{Limitations \& Future Directions}
ARCHS4 database contains data from short-read RNA-seq, and alignment of short reads to quantify transcript isoforms and their respective exons is challenging and an active area of ongoing research \citep{Steijger2013}. Inaccurate attribution of short reads might result in misleading exon counts between transcript isoforms and skewed $\psi$ levels, which limits the prediction usability in a clinical setting. 

Another possible issue arises from the fact that samples from the ARCHS4 database are heterogeneous. We speculate that this heterogeneity contributes to the variance in DCEN's performance across the different tissue types (\S~\ref{sec:tissue type}), since the experimental parameters might differ across different labs while collecting the data from these different tissues.

Although we performed batch effect removal with a standard \textit{edgeR} procedure, batch effects might remain in the data. Homogeneous datasets, as Genotype-Tissue Expression (GTEx) \footnote{https://gtexportal.org} may be used instead, but using ARCHS4 database offers advantages in 1) the diversity of tissue types (healthy and diseased) and 2) its amendable form of data to be further processed for downstream applications.
 
In the training and evaluation here, the predictions are inferred based on a universal DNA sequence assumed to be the same for all patients. Genomic variations, such as mutations and single nucleotide variants, often lead to aberrant splicing outcomes. One future direction would be to utilize DCEN to study the roles of such genomic variations in the alternative outcomes. Through DCEN's hierarchical approach of modeling whole transcripts' probabilities, it is possible to not only draw insights into how mutations can affect inclusion levels of individual splice sites \citep{Jaganathan2019} but also into the relative expression of transcripts.
 

\section{Conclusions}
We curate CAPD to benchmark learning models on alternative splicing (AS) prediction as a regression task, to facilitate future work in this key biological process. By exploiting the compositionality of discrete components, we propose DCEN to predict the AS outcome by modeling mRNA transcripts' probabilities through its constituent splice junctions' energy levels. Through our experiments on CAPD, we show that DCEN outperforms baselines and other ablation variants in predicting AS outcomes. Our work shows that deconstructing a task into a hierarchy of discrete components can improve performance in learning models. We hope that DCEN can be used in future work to study RNA regulatory factors' role in aberrant splicing events.

\section{Acknowledgments}

This work is supported by the Data Science and Artificial Intelligence Research Center (DSAIR), the School of Computer Science and Engineering at Nanyang Technological University and the Singapore National Research Foundation Investigatorship (NRF-NRFI2017-09).

\bibliographystyle{ACM-Reference-Format}
\bibliography{camera-ready}

\appendix
\section{Appendix}

\subsection{Additional Results} \label{sec:Additional results}

\begin{table}[!htbp]
    \centering
    \caption{Performance of DCEN with early-stopping versus 5x training steps.}
        \begin{tabular}{ lcc }
         \hline
         Model & Spearman Cor. & Pearson Cor. \\
         \hline
         Early-Stopped & 0.623 & 0.651  \\
         5x Training steps & 0.620 & 0.650 \\
         \hline
        \end{tabular}
\label{tab:test 5x iters}
\end{table}

\begin{figure}[H]
    \centering
    \includegraphics[width=\linewidth]{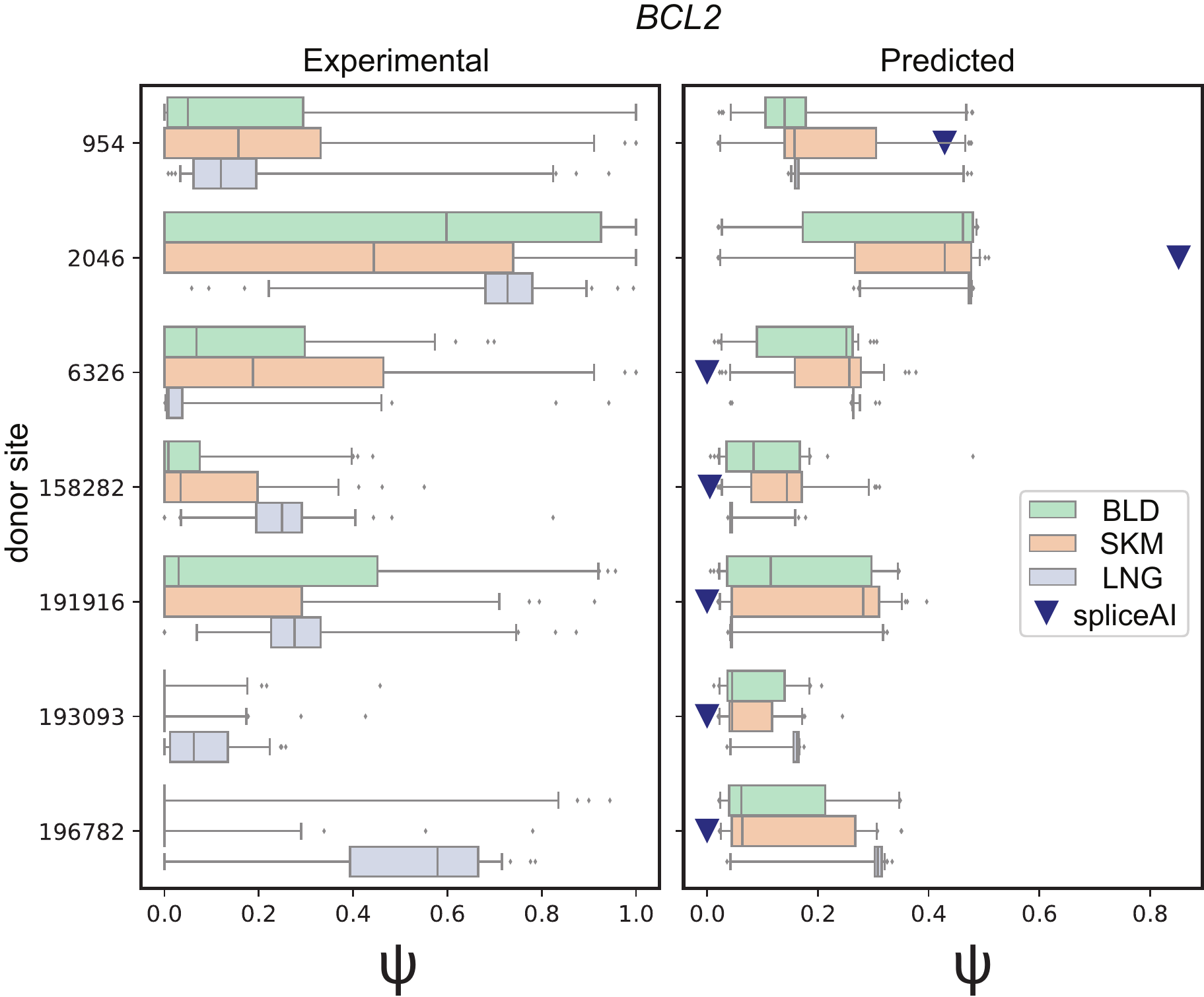}
    \caption{Comparison of the BCL2 gene donor sites' $\psi$, experimental vs predicted. Left: box plots of experimental values (grounds truths) of BCL2 donor sites in three different tissues: blood, skeletal muscle and lung, denoted as BLD, SKM and LNG respectively. Right: box plots of predicted $\psi$ for the same group of samples using DCEN; SpliceAI-cls prediction (black triangles) shown for comparison. Please note that SpliceAI-cls prediction uses sequence information only and the prediction will be the same for every sample. Donor site coordinates (y axis) are counted relative to the genomic coordinate of the first exon start of BCL2. Data from 50 samples from each tissue are shown.}
    \label{fig:BCL2don}
\end{figure}

\newpage
\subsection{Proof} \label{sec:proof}
\begin{theorem} \label{theorem:appendix energy levels to probabilities}
Given the energy levels of all the possible discrete states of a system, the probability of a particular state $T_i$ is the softmax output of its energy $E_{T_i}$ with respect to those of all other possible states in the system, i.e.,

\begin{equation} \label{eq:energy levels to probabilities}
P_i = \frac{ \exp (- E_{T_i} ) }{ \sum_{j} \exp (- E_{T_j}  ) } = \mathrm{Softmax}_i( E )
\end{equation}

\end{theorem}

\begin{proof}
From Boltzmann distribution, the probability that a system takes on a particular state ($x$) can be expressed as:

\begin{equation} \label{eq:Boltzmann distribution}
    \begin{aligned}
p_{\theta} (x) 
& = \frac{ \exp{ (- E_{\theta} (x)} ) }{ Z (\theta) } \\
& = \frac{ h(x) }{ Z (\theta) }
    \end{aligned}
\end{equation}

where 

\begin{equation} \label{eq:Boltzmann distribution}
    \begin{aligned}
Z (\theta)
& = \int \exp{ (- E_{\theta} (x) )} ~ dx \\
& = \int h(x) ~dx
    \end{aligned}
\end{equation}

is known as the partition function.

Since the probabilities of all possible states sum to 1, we have
\begin{equation}
1 = \E_{x \sim p_{\theta}} [1] = \sum_{x} \frac{ h(x) }{ Z }
\end{equation}

which gives
\begin{equation}
Z = \sum_{x} h(x) ~.
\end{equation} 

Through importance sampling with another probability distribution $q$, we can express $Z$ as

\begin{equation} 
    \begin{aligned}
Z 
& = \sum_{x} \frac{ h(x) }{ q(x) } q(x) \\
& = \frac{1}{n} \sum_{i} \frac{ h(x_i) }{ q(x_i) } ~~~,~~~~~ x_i \sim q  ~.\\
    \end{aligned}
\end{equation} 

Using an uniform discrete distribution as $q$ where all $k$ possible states ($x_i$) have the same probability $q(x_i) = (1/k)$, we get

\begin{equation} \label{eq:partition function}
    \begin{aligned}
Z 
& = \frac{1}{k} \sum_{i} \frac{ h(x_i) }{ (1/k) } \\
& = \sum_{i} h(x_i) 
    \end{aligned}
\end{equation} 

Combining Eq.~(\ref{eq:Boltzmann distribution}) and (\ref{eq:partition function}), this gives

\begin{equation} \label{eq:softmax probability}
    \begin{aligned}
p_{\theta} (x_i) 
& = \frac{ h(x_i) }{ \sum_{j} h(x_j) } \\
& = \frac{ \exp{ (- E_{\theta} (x_i)} ) }{ \sum_{j} \exp{ (- E_{\theta} (x_j)} ) } \\
& = \mathrm{Softmax}_i( E )
    \end{aligned}
\end{equation}

\end{proof}

\end{document}